\documentclass[11pt,a4paper]{article}
\usepackage[utf8]{inputenc}
\usepackage{amsfonts}
\usepackage{amscd}
\usepackage{amsmath,bm}
\usepackage{theorem}
\usepackage{mathrsfs}
\usepackage{fancybox,amssymb}
\usepackage[english]{babel}
\usepackage{graphicx}
\usepackage{enumerate}
\usepackage{float}
\usepackage{epstopdf}
\usepackage{authblk}
\usepackage{marvosym}
\usepackage{color}
\usepackage[hang,flushmargin]{footmisc}

\newcommand{\R}{\mathbb{R}}
\newcommand{\N}{\mathbb{N}}
\newcommand{\mP}{\mathbb{P}}
\newcommand{\mE}{\mathbb{E}}
\newcommand{\md}{\,{\rm d}}

\newcommand{\one}{1\mkern-5mu{\hbox{\rm I}}}

\theoremstyle{break}
\theorembodyfont{}
\newtheorem{Def}{Definition}[section]

\newtheorem{Bem}[Def]{Remark}

\newtheorem{Lem}[Def]{Lemma}
\newtheorem{Prop}[Def]{Proposition}

\newtheorem{Bsp}[Def]{Example}

\makeatletter
\newenvironment{proof}{\noindent{\textit{Proof:}}}{%
\unskip\nobreak\hfil\penalty50\hskip1em\null\nobreak
$\Box$
\parfillskip=\z@\finalhyphendemerits=0\endgraf\bigskip}

\let\oldendBsp\endBsp
\def\endBsp{\unskip\nobreak\hfil\penalty50\hskip1em\null\nobreak\hfil%
$\blacksquare$\parfillskip=\z@\finalhyphendemerits=0\endgraf\oldendBsp}
\let\oldendBem\endBem
\def\endBem{\unskip\nobreak\hfil\penalty50\hskip1em\null\nobreak\hfil%
$\blacksquare$\parfillskip=\z@\finalhyphendemerits=0\endgraf\oldendBem}
\makeatother
\date{}

\title{Transforming public pensions: A mixed scheme with a credit granted by the state}

\author{M. Carmen Boado-Penas\footnote{University of Liverpool}\quad\quad Julia Eisenberg\footnote{TU Wien, jeisenbe@fam.tuwien.ac.at, \;\Letter}\quad\quad Ralf Korn\footnote{University of Kaiserslautern and Fraunhofer ITWM}}

\begin{document} 
\maketitle

\begin{abstract}\noindent
Birth rates have dramatically decreased and, with continuous improvements in life expectancy, pension expenditure is on an irreversibly increasing path. This will raise serious concerns for the sustainability of the public pension systems usually financed on a pay-as-you-go (PAYG) basis where current contributions cover current pension expenditure. With this in mind, the aim of this paper is to propose a mixed pension system that consists of a combination of a classical PAYG scheme and an increase of the contribution rate invested in a funding scheme. The investment of the funding part is designed so that the PAYG pension system is financially sustainable at a particular level of probability and at the same time provide some gains to individuals. In this sense, we make the individuals be an active part to face the demographic risks inherent in the PAYG and  re-establish its financial sustainability.
	\vspace{6pt}
	\settowidth\labelwidth{{\it 2010 Mathematical Subject Classification: }}%
	\par\noindent {\it 2010 Mathematical Subject Classification: }%
	\rlap{Primary}\phantom{Secondary}
	93E20\newline\null\hskip\labelwidth
	Secondary 91B30, 60K10
\end{abstract}

\section{Introduction}
The decline in fertility rates, the increase in longevity and the current forecasts for the ageing of the baby-boom generation all point to a substantial increase in the age dependency ratio, and this will raise serious concerns for the sustainability of PAYG pension systems. In particular, the life expectancy at birth is expected to increase by 5.3 for males and 5.1 years for females when comparing 2016 with 2070 (European Commission \cite{Ageingreport}).
This is a worldwide problem, and consequently, many European countries (European Commission \cite{GreenPaper,WhitePaper}) have already carried out some parametric reforms, or even structural reforms, of their pension systems.
\\In Europe, the common trend of the pension crisis is a wave of parametric adjustments including countries, among others, France, Greece, Hungary, Romania and Spain, see Whitehouse \cite{Whitehousea,Whitehouseb} and OECD \cite{oecd2011,oecd2012,oecd2013,oecd2017}. Among the major changes in pension reform is the introduction of what is known as Notional Defined Contribution pension schemes (NDCs), first developed in the 1990s in countries such as Italy, Latvia, Poland and Sweden. NDCs aim at reproducing the logic of a funded defined contribution plan - (accumulated contributions are converted into a life pension annuity taking into account life expectancy according to standard actuarial practice) albeit, under a PAYG framework.
\\In Latin America, since the 1980's, most of the countries in the region made structural reforms replacing completely or partially their PAYG system with programmes containing a fully funded component of individually capitalised accounts (Rofman et al. \cite{Rofman}). As a result, a transfer of financial market and volatility risk from the state to the individual happened.
\\The PAYG rate of return can be lower than the rate of return of funding schemes, especially in countries where the working population is not growing. In this case, the individual might consider that there is an implicit cost equivalent to the difference in return; see Robalino and Bodor \cite{Robalino} and Vald\'{e}s-Prieto \cite{Valdes} for taxes implicit to PAYG schemes. However, the high-variability of the funding rate of return makes the choice between PAYG and funding less obvious and there might be advantages of mixing PAYG and funded schemes (De Menil et al. \cite{Menil}, Persson \cite{Persson}). Also, PAYG is a useful social security financing technique which ensures income redistribution at both the inter and intra-generational levels.
\\On the other hand, Fajnzylber and Robalino \cite{Fajnzylber} state that the transition from a PAYG to a fully funded scheme has a high transition cost where current contributors pay twice: first to finance their own retirement capital and second to finance pension benefits of current retirees.
Currently, countries, such as Australia, Canada, Norway, Sweden, Latvia and Poland, amongst others, combine funded and PAYG elements within the mandatory pension system.  
These mixed systems have been advocated, particularly by the World Bank, as a practical way to reconcile the higher financial market returns compared with GDP growth with the costs of a scheme with a greater funded element.\\
The academic literature has already focused on mixing pension systems in the past years. Merton \cite{Merton} studies this problem under a general equilibrium model and concludes that investors cannot achieve an optimal portfolio allocation of their savings. Matsen and Thogersen \cite{Matsen}, Knell \cite{Knell} and Guigou et al. \cite{Guigou} calculate the optimal split between funded and unfunded pension savings. Dutta et al. \cite{Dutta} show that in a mean-variance framework the mix of funded and unfunded is desirable because it enables risk diversification. Devolder and Melis \cite{Devolder} analyse the portfolio allocation problem between various financial assets and PAYG and obtain constant portfolio allocations. Alonso-Garc\'ia and Devolder \cite{Alonso} study the cohort's optimal mix between funded and unfunded schemes and whether there are diversification benefits for the specific context of a notional defined contribution scheme with a constant contribution rate.\medskip
\\
On the other hand, Robalino and Bodor \cite{Robalino} propose the use of government indexed bonds - first introduced by Buchanan \cite{Buchanan} to support the sustainability of PAYG schemes. In particular, Robalino and Bodor \cite{Robalino} analyse the case of a cash-surplus of the pension fund being invested every six or twelve months in GDP indexed government bonds. However this approach does not guarantee the financial sustainability of the system, and the application is done to notional account systems. Other papers, Auerbach and Kotlikoff \cite{Auerbach}, Sinn \cite{Sinn}, Palacios and Sinn \cite{Palacios}, also provide some discussion about the use of government bonds.
\\With this in mind, our research extends the literature on mixed pension systems and we analyse two types of strategies under which 
an investment into a defined contribution (DC) funded scheme can restore the financial sustainability of the system and at the same time provide, on expectation some gains to the individual. In this sense, we also make the individuals be an active part to re-establish the risks (mainly demographic and longevity) inherent in the PAYG.\medskip
\\
Following this introduction, the next section describes the model together with the assumptions used. In our proposed modelling framework we take on the view of a prototypical customer (PC) that earns the average salary and has the average age. This can be regarded as a standardization approach comparable to the representative agent in finance. Further, our model is based on a credit granted by the state to the PC in the sense that the state covers the deficit as soon as it happens and the PC will be paying to the state in the future. In Section 3, some variants are presented where the PC invests the corresponding money into a fund and needs to repay the credit either at the end of the year or after a longer specified period. A comparison between the annual approach and the long-term approach is given along. In Section 4, we assume that the individuals need to transfer any excess of return above some particular level that needs to be optimised. We compare the results with the corresponding discrete time approach of Section 3. Section 5 concludes and makes suggestions for further research.
\section{The Model}

This section describes the model and some of the assumptions made in the paper so that the classical PAYG can be transformed into a mixed pension system socially accepted by the individuals with contributions paid into a PAYG and a funding scheme.

\noindent
We consider a prototypical contributor (PC), i.e.\ \textit{the average contributor}, with an average salary and average salary increases. This PC has to contribute an amount $C_0$ -- expressed in percentage of his salary -- at time $t=0$ into the PAYG.\\
With the aim of restoring the financial sustainability of the PAYG scheme an automatic balancing mechanism will be triggered if the income from contributions is not sufficient to pay for the pension expenditure. This mechanism is based on increases of the contribution rate paid by the PC.\\
Let assume that the state anticipates the deficit and it is known that the contributions that make the PAYG system sustainable over time are given by
\begin{equation*}
C = (C_1,...,C_T)
\end{equation*}
and
\[
C_j > C_0,\; j=1,...,T\ .
\]
Considering that the state targets to transform the classical PAYG scheme and is taking over the payment of the differences $C_j-C_0$ for $j=1,...,T$. However, these payments represent some kind of a credit because, as soon as the deficit occurs the PC has to invest some pre-specified amount of money in a fund, in addition to the regular contribution of $C_0$ to PAYG. 
A certain part of the return on investment will be used to cover the debt amount of $C_j-C_0$ in the (near) future. The remaining capital belongs to the PC and could be used for instance for the retirement phase. The government will bear the risk of not full debt repayment if the return on investment is not sufficient to cover the deficit for one particular year.

\noindent
The following sections specify two types of investment strategies depending on the amount invested and the time horizon for the debt repayment. We calculate the probability of full payback of the debt, the expected loss of the state and the expected gains of the PC.
\\In the following we act on some filtered space $(\Omega,\mathfrak A,(\mathcal F_t)_{t\geq 0},\mP)$ where $\mathcal F$ is assumed to be the right-continuous filtration generated by a given geometric Brownian motion $F$.
\section{Lump Sum Repayment\label{sec:ls}}
In this section, we consider an annual fund investment approach where the increase of the contribution paid by the PC is invested for a year. Depending on the the actual way how the risk of a shortfall is shared between PC and the state, we suggest two different models.

\subsection{Payback first}
We assume that \textbf{at time $t=0$} it is agreed that the contribution to PAYG of the PC is fixed to be $C_0$ for the next $T+1$ years. In return, the PC invests an amount of $\alpha (C_j-C_0)$ into a fund with value dynamics given by
\[
F_t = F_0 e^{\mu t + \sigma W_t} \ 
\]
at the beginning of year $j+1$ (i.e. at time $j$) for some $\alpha>0$. Then, at time $j+1$, the PC pays $C_j-C_0$ to the state if 
\[
\alpha (C_j-C_0) e^{\mu + \sigma (W_{j+1}-W_{j})} \ge C_j -C_0 \ .
\]
If this is the case, then the remaining value of the fund position stays with the PC. In the other case, the state receives the full fund position and takes over the loss of
\[
L_j:=(C_j-C_0)\cdot \left(1-\alpha e^{\mu + \sigma (W_{j+1}-W_{j})}\right)\ .
\]
This procedure is repeated every year.\footnote{If the PC invests at time $j$ - in line with the necessary increase of the PAYG contribution - and pays back at time $j+1$, then the procedure also includes a credit for a duration of one year.}

\noindent Certainly, the state wants to keep the probability of losses small, while the whole procedure is only attractive for the PC if there is a possible gain at least on the level of expectations. It is thus clear that the fund characteristics play an important role. We, therefore, summarize some properties in the following proposition. 
\begin{Prop}\label{prop:1}
For a value of $\alpha>0$ and $C_j>C_0 >0$, we obtain:\\[1mm]
\noindent
a) The probability that the full payment is made to the state at time $j$ is given by
\begin{align*}
P_j:=\mP\big[\text{Full payment at time }j+1\big] &= \mP\big[\alpha e^{\mu + \sigma (W_{j+1}-W_{j}}) \ge 1\big] \\
 &= \Phi\left(\frac{\mu+\ln(\alpha)}{\sigma}\right) \;.
\end{align*}
b) The expected loss $\mE[L_j]$ of the state at time $j$ is given by
\begin{align*}
\mE[L_j] &= \mE\left[(C_j-C_0)\left(1 - \alpha e^{\mu + \sigma (W_{j+1}-W_{j})}\right)^+\right]
\\&= (C_j-C_0) \bigg\{\Phi\left(-\frac{\mu+\ln(\alpha)}{\sigma}\right) 
 - \alpha e^{\mu+\frac{\sigma^2}2} \Phi\left(-\frac{\mu+\sigma^2+\ln(\alpha)}{\sigma}\right)\bigg\} \ .
 \end{align*}
c) The expected gain $\mE[G_j]$ of the PC at time $j$ is given by
\begin{align*}
\mE[G_j] &= \mE\left[(C_j-C_0)\left( \alpha e^{\mu + \sigma (W_{j+1}-W_{j}})-1\right)^+\right] 
\\&= (C_j-C_0) \bigg\{\alpha e^{\mu+\frac{\sigma^2}2} \Phi\left(\frac{\mu+\sigma^2+\ln(\alpha)}{\sigma}\right)  
  -\Phi\left(\frac{\mu+\ln(\alpha )}{\sigma}\right)\bigg\} \ .
\end{align*}
\end{Prop}
\begin{proof}
Assertion a) is implied by the fact that $W_{j+1}-W_{j}$ is standard normally distributed. \\[1mm]
\noindent
The proof of Assertions b) and c) consists of explicit calculations which are totally similar to those in the proof of the Black-Scholes formula. Thus, we only give the proof of Assertion b): \\[1mm]
\noindent
Noting again the normal distribution property of $W_{j+1}-W_{j}$, we obtain
\begin{align*}
\mE[L_j] &= (C_j-C_0) \bigg(\frac{1}{\sqrt{2\pi}}\int_{-\infty}^{-\frac{\mu+\ln(\alpha)}{\sigma}}{e^{-\frac{y^2}{2}}\md y}  
 -\frac{\alpha}{\sqrt{2\pi}}\int_{-\infty}^{-\frac{\mu+\ln(\alpha)}{\sigma}}{e^{\mu+\sigma y-\frac{y^2}{2}}\md y}\bigg)  \\
&= (C_j-C_0) \left\{\Phi\left(-\frac{\mu+\ln(\alpha)}{\sigma}\right) - \alpha e^{\mu+\frac{\sigma^2}2} \Phi\left(-\frac{\mu+\sigma^2+\ln(\alpha)}{\sigma}\right)\right\}\;.
\end{align*}
\end{proof}
The probability of being able to fully pay $C_j-C_0$ back to the state at time $j+1$ is of course increasing in $\alpha$. The same is true for the expected final fund position of the PC given that $C_j-C_0$ (or the whole position in case it is not sufficient for the full payment) has been already paid to the state. However, a time horizon of just one year is very short. Consider for instance the case of 
\[
\alpha = 1\ ,
\]
i.e. one invests exactly the debt amount $C_j-C_0$ at the beginning of year $j+1$ into the fund. Then, part a) of Proposition \ref{prop:1} implies 
that the probability for obtaining $C_j-C_0$ at time $j$ is given by
\begin{equation*}
   \mP\left[(C_j-C_0) e^{\mu +\sigma (W_{j+1}-W_{j})}> C_j-C_0\right]=\Phi\left(\frac{\mu}{\sigma}\right)\ .
\end{equation*}

\noindent Note that for a fund with optimistic parameters $\mu=0.04, \sigma=0.2$ this probability is approximately $58\%$ which might be considered as not satisfactory enough. 

\noindent
In order to reduce the probability of a shortfall (i.e. not full debt repayment by the PC) significantly, the state should require
\begin{equation*}
\alpha \gg 1\ ,
\end{equation*} 
i.e. the PC should contribute $C_0$ plus a fund investment such that the resulting sum is well-above the actually needed contribution of $C_j$ at time $j$. At the same time, only a strong requirement for a small probability of shortfall will lead to building up a significant amount of money in a fund for the PC, as shown in the following example.

\begin{Bsp}
We illustrate the issues raised above by looking at two settings of values for the fund parameters:\\[5mm]
\noindent
a) \textbf{A standard fund}:  $\mu=0.04,\ \sigma=0.2$\\[1mm]
In this setting, the probabilities for a full payment of $C_1-C_0$ at time $1$ as a function of $\alpha$, denoted by $P_1$, are computed. Further, we also calculate the expected loss for the state, $\mE[L_1]$, and the expected gain of the PC after paying to the state, $\mE[G_1]$. In addition we look at the following two quantities:
\begin{align*}
&\mE[B_1]:= \mE\left[(C_1-C_0)\left( \alpha e^{\mu + \sigma W_1}-1\right)\right]\;,
\\&\mE[G^{A}_1]:=\mE[G_1]- (\alpha-1)(C_1-C_0)\;.
\end{align*}
It means $\mE[B_1]$ is the expected fund position after one year if always the full value $C_1-C_0$ has to be paid back by the PC, and $\mE[G^{A}_1]$ describes expected gain for the PC minus the additional investment of $(\alpha-1)(C_1-C_0)$. All above quantities for $C_0=1,\ C_1=1.1$ can be found in Table \ref{tab:loss_net} below.
\begin{table}[h]
	\centering
		\begin{tabular}{c|ccccccc}
			$\alpha $ & 1 & 1.05 & 1.1 & 1.15 & 
			1.25 & 2 & 3 \\
			\hline
			$P_1$ & 0.58 & 0.67 & 0.75 & 0.82 & 0.91 & $>$ 0.99 & 1.00 \\
			$\mE[L_1]$ & 0.005 & 0.004 & 0.003 & 0.002 & 
			0.001 & $<10^{-4}$ & $<10^{-4}$\\
			$\mE[G_1]$ & 0.0117 & 0.015 & 0.020 & 0.024 & 
			0.034 & 0.112 & 0.219\\
			$\mE[B_1]$& 0.006 & 0.011 & 0.017 & 0.022 & 
			0.033 & 0.112 & 0.219\\
			$\mE[G^{A}_1]$& 0.0117 & 0.0104& 0.0095 & 0.009&
			0.0085 & 0.0124 & 0.0186
		\end{tabular}
		\caption{The key quantities of the scheme as functions of $\alpha$ for $C_0=1$, $C_1=1.1$.}
		\label{tab:loss_net}
\end{table}

\noindent
While the first four rows in Table \ref{tab:loss_net} show an obvious behaviour (i.e. the full payback probability and the gains increase and the expected loss of the state decreases with increasing $\alpha$), it is actually the last row that should be noticeable. The expected difference between the net gains and the additionally invested amount first decrease with increasing $\alpha$ and then increase for large values of $\alpha$. 
In particular, for values of $1\le \alpha \le 1.25$ there is no real incentive for PC to exceed $\alpha=1$. The reason for that is that, initially, $\mE[G_1]$ does not grow fast enough with increases in $\alpha$. There is an initially large benefit for PC as the state takes over some significant investment risk via accepting an expected loss of $5\%$ per unit of money given as a credit. If the state accepts an expected loss of $10\%$ per unit of money (this is the case for $\alpha=0.9$) then PC has already realised a sure gain of $1\%$ plus a small expected net wealth after paying back $C_j-C_0$. In total, this leads to $\mE[G_1]=0.0157$. However, such a (partially) safe gain can only be interpreted as a subsidy granted by the state to motivate fund investment. \\[2mm]
\noindent
The effect of the state taking over the shortfall risk becomes insignificant with increasing values of $\alpha$. Even more, for very large values of $\alpha$ the probability that we end up with a zero fund position after paying back vanishes, and the fund investment can unfold its full potential.\\[2mm] 
\noindent
If the state required a probability for a full credit payback of $0.9$, $0.95$ or $0.99$ then corresponding values of $\alpha$ are given by $1.24$, $1.34$ and $1.53$. Thus, to avoid a loss with a probability of $99\%$, the PC needs to invest $53\%$ more in the fund than when paying the increase in the contribution directly into the PAYG. On the positive side, the PC obtains an expected net fund wealth of $6.25\%$ for additionally paying in $5.3\%$ in our example\\[5mm]
\noindent
b) \textbf{A very diversified fund}:  $\mu=0.04,\ \sigma=0.1$\\[1mm]
In this setting, the probability of a full credit payback, $P_1$, as a function of $\alpha$ are given in Table \ref{tab:loss_prob2}. 
Note that the probability of the full payment and the expected loss are better than in the case of the fund with a higher volatility. However, as shown in the table, the expected net gains for PC are slightly smaller. 
\begin{table}[h]
	\centering
		\begin{tabular}{c|cccccc}
			$\alpha$ & 1 & 1.05 & 1.1 & 1.15 & 1.2 & 1.25 \\
			\hline
			$P_1$ & 0.66 & 0.81 & 0.91 & 0.96 & 0.99& 0.996\\
			$\mE[L_1]$ & 0.002 & 0.001 & $4\cdot 10^{-4}$ & $10^{-4}$ & $<10^{-4}$& $<10^{-4}$ \\
			$\mE[G_1]$ & 0.007 & 0.011 & 0.015 & 0.020 & 0.026& 0.031
		\end{tabular}
		\caption{The probability for the full debt repayment $P_1$, the expected loss of the state $\mE[L_1]$ and the expected gain of the PC $\mE[G_1]$ as functions of $\alpha$.}
		\label{tab:loss_prob2}
\end{table}
\end{Bsp}
\subsection{A payback variant above a certain return} 
We will here assume a different variant of paying (back) the difference $C_j-C_0$ after the investment of $\alpha (C_j-C_0)$ has been done and the investment result has realised. This variant will be considered in detail in Section \ref{sec:opti_approach} in a continuous-time framework that allows a more flexible version.\\[2mm]
\noindent
Let us therefore assume that the state agrees that the PC can keep at least a certain return of $b\ge -1$ \big(i.e. an amount of $(1+b)\alpha(C_j-C_0)$\big) at the end of the investment period. Whatever exceeds this - but does not exceed $C_j-C_0$, goes to the state. If, however, the investment result is less than $C_j-C_0$ then the possibly remaining part of $C_j-C_0$ will be taken over by the state. 
We denote by $R_j$ the amount of money that the PC can keep and by $D_j$ the amount of money to be used for the debt repayment at time $j$. This yields
\begin{align*}
R_j&=\begin{cases}
    \alpha(C_j-C_0)(1+b) &\mbox{: if $e^{\mu+\sigma W_1}>1+b$}, \\\alpha(C_j-C_0)e^{\mu+\sigma W_1} &\mbox{: else},
    \end{cases}
\\D_j&= \alpha(C_j-C_0)\left(e^{\mu+\sigma W_1}-(1+b)\right)^+\ .
\end{align*}
The explicit form of $D_j$ directly implies that we either need a value of $b$ close to $-1$ and a value of $\alpha$ close to $1$ or a very large value of $\alpha$ in the case of $b\approx 0$ to obtain a significantly high probability for a full payback of $C_j-C_0$ to the state at time $j+1$. Indeed, we have 
\begin{align*}
\mP\left[D_1\ge C_1-C_0\right]&=\mP\left[W_1\ge \frac{\ln \left(1+b +1/\alpha\right)-\mu}{\sigma}\right]\\
&=\Phi\left(\frac{\mu-\ln \left(1+b +1/\alpha\right)}{\sigma}\right)\;.
\end{align*}
However, in the case of $b = 0$ we will never obtain a probability for a full debt repayment higher than approximately $\Phi\Big(\frac{\mu}\sigma\Big)$, which will even decrease for bigger values of $b$. Further, a non-positive value of $b$ can only be preferable for the PC in the case of $\alpha<1$, but then the probability for a full debt repayment to the state decreases. \\[2mm]
\noindent
The probability of the full repayment for 1-year investment and the standard fund with $\mu=0.04, \sigma = 0.2$ are summarised in Table \ref{tab:loss_prob3}.
\begin{table}[h]
	\centering
		\begin{tabular}{r|cccccc}
			$\alpha$ & 0.8 & 0.9 & 1 & 
			1.25 & 2 & 10 \\
			 \mbox{$b$}\hspace{0.6cm} & & & & & &\\
			\hline
			0.02 & $<$0.001 & $<$0.001 & $<$0.001 & 
			0.003 & 0.0291 &0.357 \\
			0 & $<$0.001 & $<$0.001 & $<$0.001 & 
			0.003 & 0.0338 &0.391\\
			-0.5 & 0.005 & 0.014 & 0.034 &  
			0.108& 0.5793 &0.997 \\
			-0.75 & 0.034 & 0.090 & 0.018 & 
			0.421& 0.9493 &1\\
			-0.9 & 0.097	& 0.224 &	0.391 &		
			0.707&	0.9971&1\\
			-0.95 & 0.133 &	0.292 &	0.482 &	
			0.794 &	0.9993 &1 \\
			-1 & 0.180 &	0.372 &	0.579 &	
			0.867 &	0.9999 &1
		\end{tabular}
		\caption{Probability for the full debt repayment as a function of $\alpha$ and $b$.}
		\label{tab:loss_prob3}
\end{table}
\noindent
The values obtained in Table \ref{tab:loss_prob3} show some clear indications: 
\begin{itemize}
 \item Even in the case of $b=-1$, a high probability (i.e. one above $80\%$) such that this investment result exceeds $C_1-C_0$ is only obtained for the cases of $\alpha \ge 1.25 $. In the latter case, the PC is better off by paying the difference $C_1-C_0$ directly to PAYG. 
 \item For values of $\alpha < 1$ the payback probabilities are far from being satisfactory for the state. 
\end{itemize}
To get a motivation for the state to allow a positive value of $b$ which also makes this kind of investment attractive for the PC, we calculate 
\begin{align*}
  \mE[D_1]&= \alpha(C_1-C_0)\bigg\{e^{\mu+\frac{\sigma^2}2} \Phi\left(\frac{\mu+\sigma^2-\ln\left(1+b\right)}{\sigma}\right)  
\\&\quad{}	 - \Phi\left(\frac{\mu-\ln\left(1+b\right)}{\sigma}\right) (1+b) \bigg\}\;,\\
	\mE[R_1]&= \alpha(C_1-C_0)\bigg\{e^{\mu+\frac{\sigma^2}2} \Phi\left(-\frac{\mu+\sigma^2-\ln\left(1+b\right)}{\sigma}\right)  \\
	& \quad {}+ \Phi\left(\frac{\mu-\ln\left(1+b\right)}{\sigma}\right) (1+b) \bigg\}\;.
\end{align*}
With high values of $\alpha$, we can see that at least in expectation both the state and the PC are winning. To understand this, note that the final wealth of the PC only has to exceed $(\alpha-1)(C_1-C_0)$ as without fund investment the PC has to pay $C_1-C_0$ in addition to $C_0$ anyway.  
Table \ref{tab:loss_probalpha} shows that there is a potential for both the PC and the state to gain in expectation. However, both have to take their shortfall risks into account, and it is by far not clear if the PC can afford such a high investment.
\begin{table}[h]
	\centering
		\begin{tabular}{r|ccc|ccc}
			$\alpha$   & 10   & 10   & 10    & 20   & 20   & 20 \\
			$b$        &0.030 &0.005 &-0.070 &0.153 &0.100 &0.009\\	
			\hline
$\mE[D_1]$ \hspace{4mm}& 0.1  &0.114 &0.162  &0.1   &0.137 &0.224 \\
$\mE[R_1]$ \hspace{4mm}&0.962 &0.948 &0.9    &2.204 &1.987 &1.9\\
		\end{tabular}		
		\caption{The expected payment to the state and the expected net fund value of the PC after (partial) payback as functions of $\alpha$ and $b$.}
		\label{tab:loss_probalpha}
\end{table}
\subsection{Variant B: A granted credit}
As a second approach, the state again pays the difference of $C_j-C_0$ at times $j=1,...,T$, where for the moment we assume that the values $C_j$ are deterministic. However, the state considers these payments as a credit with zero interest rate\footnote{\textcolor{black}{This is a reasonable assumption given the current ultra low interest rate environment in Europe. \color{black}}} to the PC. This construction allows the PC to invest money in a fund for a longer time and delays the full payback to time $T$. The potential gains from the fund investment are then used 
\begin{itemize}
	\item to pay back the credit and
	\item to build up money as reserves for the own future. 
\end{itemize}
Thus, it remains to find the optimal additional amount of money, that the PC invests into the fund, and the strategy to pay back the credits granted by the state. As in Variant A, we consider the following strategy:\\[2mm] 
\noindent
At time $j-1$ the PC contributes $C_0 + \alpha (C_j-C_0)$ with $\alpha > 0$, $j=1,...,T$. $C_0$ directly goes to PAYG while $\alpha (C_j-C_0)$ is invested into the fund. Thus, at time $j$ the PC has a fund position $F_j$ (before payment of the new contribution), where
\begin{align*}
  &F_j := \alpha \sum_{k=1}^{j}{\left((C_{j-k+1}-C_0)e^{\mu k + \sigma (W_{j+1}-W_{j-k+1})}\right)}
 \\&\mE\left[F_j\right] = \alpha \sum_{k=1}^{j}{(C_{j-k+1}-C_0)e^{(\mu + \frac{\sigma^2}{2})k}}                  
\end{align*}
for $j=1,...,T$. In the special case of $C_j=\bar{C}$ for all $j=1,...,T$, we have
\begin{equation*}
 \mE\left[F_j\right] = \alpha \sum_{k=1}^{j}{(C_{j-k+1}-C_0)e^{(\mu + \frac{1}{2}\sigma^2)k}} 
                     = \alpha (\bar{C}-C_0) e^{\mu+\frac{1}{2}\sigma^2}\frac{1-e^{(\mu + \frac{1}{2}\sigma^2)j}}{1-e^{\mu + \frac{1}{2}\sigma^2}} 
\end{equation*}
We assume that the full credit sum, defined as
\[
C_{\rm total}:=\sum_{j=1}^T{(C_j-C_0)}
\]
will be paid to the state at time $T$. \\[1mm] 
\noindent
We thus arrive at
\begin{Lem}
For an $\alpha > 0$ and the above investment strategy into the fund, we have: \\[1mm]
\noindent
a) In the case of payments of $\alpha (C_j-C_0)$ at time $j$, $j=1,...,T$ into the fund, the credit can be fully paid back in expectation at time $T$ if we choose
\begin{equation*}
\alpha =\alpha^* := \frac{\sum_{j=1}^T{(C_j-C_0)}}{\sum_{j=1}^{T}{(C_{j}-C_0)e^{(\mu + \frac{1}{2}\sigma^2)(T+1-j)}}} \ .
\end{equation*}
In the particular case of $C_j=\bar{C}$ for all $j=1,...,T$, this simplifies to 
\begin{equation*}
\alpha= \alpha^*=T \frac{1}{e^{\mu + \frac{1}{2}\sigma^2}}\frac{1-e^{\mu + \frac{\sigma^2}{2}}}{1-e^{(\mu + \frac{1}{2}\sigma^2)T}} \ .
\end{equation*}
b) The expected gain for the PC in nominal values by comparing the credit sum $\sum{(C_j-C_0)}$ with the invested amount $\alpha^*\sum{(C_j-C_0)}$ is strictly positive as we have $\mu +\frac{\sigma^2}{2} > 0$.
\end{Lem}
\begin{proof}
a) follows from the required equality between the sum of (expected) payments by the state and $\mE(F_T)$ which then has to be solved for $\alpha^* C_0$. \\b) is obvious.
\end{proof}
In contrast to the one period setting of Variant A, we cannot easily calculate the probability of a loss by the state or the expected net fund value for the PC at time $T$. The reason is a well-known fact that the distribution of the sum of log-normally distributed random variables is not explicitly known. We therefore illustrate the performance of this strategy and its differences to Variant A in the next section.
\subsection{Comparison of Variant A and Variant B}
As we do not have explicit distributional results for Variant B, we present a numerical example where we choose $T=10$ and also $C_j=\bar{C}=1.1\cdot C_0$. We then simulate $10,000$ realisations of the fund performance and estimate the probability of a shortfall/loss (i.e. the event of $F_T< T\cdot(\bar{C}-C_0)$), the expected loss for the state and the expected net wealth for the PC that remains in the fund at time $T$. We again consider the two different fund parameter sets already used for illustrating Variant A. We especially want to compute (or more precisely, estimate via Monte Carlo simulation)
\begin{align*}
P_{\rm shortfall}&:=\mP\left[F_T\le T\cdot (\bar{C}-C_0\right] =\mP\left[F_T\le 1\right],\\
E_{\rm shortfall}&:=\mE\left[(1-F_T)^+\right], \\
E_{\rm final net fund}&:=\mE\left[(F_T-1)^+\right]\;. 
\end{align*}
For this, note first that in our setting, we can explicitly calculate
\begin{equation*}
\mE[F_T] = \alpha\cdot 0.1\cdot e^{\mu+\frac{1}{2}\sigma^2}\frac{e^{10(\mu+\frac{1}{2}\sigma^2)}}{1-e^{\mu+\frac{1}{2}\sigma^2}} \ .
\end{equation*}
By the well known relation
\[
\mE[X] = \mE[X^+]-\mE[X^-]
\] 
we can thus estimate/compute $E_{\rm final net fund}$ from $\mE[F_T]$ and the estimate for $E_{\rm total loss}$.
\subsubsection{\textbf{A standard fund:} $\mu = 0.04, \sigma=0.2$.}
\noindent
Under our above assumptions on the fund parameters and $\bar{C}-C_0=0.1\cdot C_0 = 0.1$, $T=10$ we obtain first: 
\begin{table}[H]
	\centering
		\begin{tabular}{c|cccccc}
			$\alpha$ & 1 & 1.05 & 1.1 & 1.15 & 1.2 & 1.25 \\
			\hline
			$P_{\rm shortfall}$ & 0.272 & 0.232 & 0.197 & 0.171 & 0.145& 0.123\\
			$E_{\rm shortfall}$ & 0.054 & 0.045 & 0.036 & 0.030 & 0.024& 0.020 \\
			$E_{{\rm final net fund}}$ & 0.470 & 0.530 & 0.593 & 0.657 & 0.722& 0.789
		\end{tabular}
		\caption{Probability for a shortfall, expected shortfall, expected net wealth as functions of $\alpha$.}
		\label{tab:loss_prob4}
\end{table}
\noindent
To see the time effect of investment one can e.g. compare the additional total nominal payment of $0.25$ with the net gain of $0.789$ in the case of the choice of $\alpha=1.25$. Of course, one should take interest rates into account. However, in the current ultra low interest rate environment in Europe, using nominal values is a good approximation. This is especially true in a country such as Germany with negative interest rates for short term investments. \\[2mm] 
\noindent
To compare the performance of Variant B with that of Variant A, we now consider the corresponding values when we use 
\[
\alpha=\alpha_{A,90}\approx 1.242\ ,
\]
the $\alpha$ that leads to a $90\%$ security level for a full payment to the state at each single payment time. It comes along with the probability of at least one loss of $0.35$ and a total expected loss of $0.0087$. It is clear that this loss probability is not comparable to those of Variant B. The break even value $\alpha_{B}$ such that the expected total loss of the state equals that of Variant A is given by $\alpha_B=1.44$. While it is comparably high, it leads to an expected net fund wealth of $1.047$ compared to the additional nominal payments of $0.44$.
\subsubsection{A very well diversified fund: $\mu = 0.04, \sigma=0.1$.}
\noindent
For this set of fund parameters we obtain: 
\begin{table}[h]
	\centering
		\begin{tabular}{c|cccccc}
			$\alpha$ & 1 & 1.05 & 1.1 & 1.15 & 1.2 & 1.25 \\
			\hline
			$P_{\rm shortfall}$ & 0.131 & 0.087 & 0.053 & 0.035 & 0.020 & 0.015\\
			$E_{\rm shortfall}$ & 0.012 & 0.007 & 0.004 & 0.002 & 0.001& 0.001 \\
			$E_{\rm final net fund}$ & 0.304 & 0.364 & 0.425 & 0.488 & 0.552& 0.616
		\end{tabular}
		\caption{Probability for a shortfall, expected shortfall, expected net wealth as functions of $\alpha$.}
		\label{tab:loss_prob5}
\end{table}
\noindent
As in the case of Variant A the loss figures now clearly improve while the net fund wealth performs slightly weaker due to the small variance. \\[2mm]
\noindent
Again, we can compare
\[
\alpha=\alpha_{A,90}\approx 1.0925\ ,
\]
the alpha that leads to a $90\%$ security level for a full payment to the state at each single payment time with its counter part in Variant B that leads to the same expected total loss of $0.004515$. This time, the break even value $\alpha_{B}=1.0927$ is very close to $\alpha_{A,90}$. It leads to an expected net fund wealth of $0.416$ compared to the additional nominal payments of $0.0927$.
\section{Optimisation Approach} \label{sec:opti_approach}
In this part of the paper, we assume the same credit scheme described in the previous section. However, we consider a different set of strategies that can be applied in order to repay the debt to the state. We assume that the state requires the PC to share the profits from an investment continuously in time, i.e. the PC has to transfer any excess above some level $b$ (to be optimised) into a special bank account during a $1$-year period. It means that once an increase in contribution becomes necessary, the state contractually agrees to pay the difference between the old and the new contributions in the following, say, 10 years. In return, the PC has to invest a certain amount of money at the beginning of every year so that at the end of the year the debt to the state can be repaid and the PC can get some positive return on investment.  
\\Our objective is to maximise the amount of money remaining after the debt repayment to the PC. The emphasis lies on the investment of the minimal possible amount such that the probability that the debt can be fully repaid to the state stays above some pre-specified level. Note that the PC carries the investment risk where the state has the risk that the PC will not be able to repay the debt. It means in particular that we allow for negative values of the level $b$.  
\\As before we model the price of the fund under consideration by a geometric Brownian motion
\[
F_t=F_0 e^{\mu t+\sigma W_t}\;.
\] 
We denote by $D_t(b)$ the part of the gains, depending on the chosen barrier $b$, that is transferred to the debt account and by $R_t(b)$ the part remaining to the PC for the time horizon $t$.
\\Let now $b\in[-1,1]$ be arbitrary but fixed. If the return at time $t$ compared to the initial value $F_0$ 
\[
\frac{F_t-F_0}{F_0}= e^{\mu t+\sigma W_t}-1
\]
exceeds the level $b$, the excess $\big(e^{\mu t+\sigma W_t}-1-b\big)$ is transferred to the debt bank account. Mathematically it means that the process $\{e^{\mu t+\sigma W_t}\}$ is downward-reflected at $1+b$. Since under the logarithm the reflection property remains preserved, we can consider the downward-reflection of $\{\mu t+\sigma W_t\}$ at $\ln(1+b)$. For a Brownian motion with drift we know that the reflected process, below the level $\ln(1+b)$, at time $t$ is given by 
\[
\mu t+\sigma W_t -\big(\max\limits_{0\le s\le t}\{\mu s+\sigma W_s\}-\ln(1+b)\big)^+\;.
\]
Therefore, transforming the downward-reflection back, gives the part remaining to the PC 
\begin{equation*}
R_t(b):=  e^{\mu t + \sigma W_{t} -  \left(\max \limits_{0 \leq s \le t}  \{\mu s+\sigma W_{s}\} -\ln(1+b) \right)^{+}},
\end{equation*}
confer for details for instance \cite[p.\ 77]{BS}. It remains to find the expression for the debt part, denoted by $D_t(b)$, to be accumulated on the special account.  
\begin{Lem}
Let $b$ be arbitrary but fixed. For any time $t\ge 0$ it holds
\[
D_t(b)=(1+b) \Big(\max \limits_{0 \leq s \leq t}  \{\mu s+\sigma W_{s}\}-\ln(1+b)\Big)^+\;.
\]
\end{Lem}
\begin{proof}
Recall first that $D_t(b)$ describes the excess of the return from the fund above the pre-specified boundary of $1+b$. 
\\In order to prove our claim consider an $\varepsilon>0$, define for $n\ge 1$
\begin{align*}
&T_1:=\inf\{s\ge 0:\; e^{\mu s+\sigma W_s}
=1+b+\varepsilon\}\\\quad&=\inf\{s\ge 0:\; \mu s+\sigma W_s=\ln(1+b+\varepsilon)\},
\\&T_{n+1}:=\inf\{s\ge T_{n}: e^{\ln(1+b)+\mu (s-T_n)+\sigma \big(W_{s}-W_{T_{n}}\big)}=1+b+\varepsilon\}
\\\quad &=\inf\Big\{s\ge T_{n}: \ln(1+b)+\mu (s-T_n)+\sigma \big(W_{s}-W_{T_{n}}\big)
 =\ln(1+b+\varepsilon)\Big\}
\end{align*}
and let 
\[
N_t(\varepsilon):=\sup\{n\ge 1: \; T_n\le t\}\;.
\]
It means we approximate the procedure of skimming off the return of the considered geometric Brownian motion by a process, depending on $\varepsilon$, where we first wait until $e^{\mu s+\sigma W_s}$ hits the level $1+b+\varepsilon$ and pay $\varepsilon$ immediately into the debt account. Subsequently, we wait until the considered geometric Brownian motion (now with the start value $1+b$) again hits the level $1+b+\varepsilon$ and transfer $\varepsilon$ into the debt account. The amount transferred into the debt account up to time $t$ is given by $\varepsilon N_t(\varepsilon)$.
Looking at the stopping times $T_n$, one sees that the process after withdrawals, denote it by $F^\varepsilon_t$, is given by
\[
F^\varepsilon_t= \Big(\frac{1+b}{1+b+\varepsilon}\Big)^{N_t(\varepsilon)}e^{\mu t +\sigma W_t}=\Big(1-\frac{N_t(\varepsilon)\varepsilon}{N_t(\varepsilon)(1+b+\varepsilon)}\Big)^{N_t(\varepsilon)}e^{\mu t +\sigma W_t}\;.
\]
Consider the part $\Big(1-\frac{N_t(\varepsilon)\varepsilon}{N_t(\varepsilon)(1+b+\varepsilon)}\Big)^{N_t(\varepsilon)}$. Since 
\[
\lim\limits_{n\to \infty}\Big(1-\frac{x}{n}\Big)^{n}= e^{-x}\;,
\]
we just need to consider $\lim\limits_{\varepsilon \to 0}\varepsilon N_t(\varepsilon)$. 
The structure of $T_n$ as hitting times of an arithmetic Brownian motion yields, confer for instant \cite[p.\ 95]{Karatzas}:
\[
N_t(\varepsilon)=\sup\bigg\{n\in\N: \max\limits_{0\le s\le t}\{\mu s+\sigma W_s\}-\ln(1+b)\ge n \cdot\ln\Big(\frac{1+b+\varepsilon}{1+b}\Big)\bigg\}\;.
\]
Therefore, we can conclude
\begin{align*}
\lim\limits_{\varepsilon\to 0} \varepsilon N_t(\varepsilon)&=\lim\limits_{\varepsilon\to 0} \frac{\varepsilon}{\ln\Big(\frac{1+b+\varepsilon}{1+b}\Big)}\ln\Big(\frac{1+b+\varepsilon}{1+b}\Big)N_t(\varepsilon)
\\&=(1+b)\Big(\max\limits_{0\le s\le t}\{\mu s+\sigma W_s\}-\ln(1+b)\Big)^+ \quad \mbox{a.s.}
\end{align*}
This means in particular that the process $F^{\varepsilon}_t$ converges to $R_t(b)$ as $\varepsilon\to 0$ a.s., implying that $\varepsilon N_t(\varepsilon)$ converges to $D_t(b)$ a.s. for $\varepsilon \to 0$. 
Therefore:
\begin{align*}
D_t(b)=\lim\limits_{\varepsilon\to 0} \varepsilon N_t(\varepsilon)=(1+b)\Big(\max\limits_{0\le s\le t}\{\mu s+\sigma W_s\}-\ln(1+b)\Big)^+\;.
\end{align*}
\end{proof}
Like in the previous section, we assume that the initial investment is given by $F_0=\alpha (C_1-C_0)$ with some positive $\alpha$, and require for every period
\begin{equation}
\mP\Big[F_0 D_t(b)\ge C_1-C_0\Big]\ge p \;\Leftrightarrow\; \mP\Big[ D_t(b)\ge \frac 1\alpha\Big]\ge p,\label{credibility}
\end{equation}
for some given $p\in[0,1]$. This credibility condition ensures that the level $b$ is not chosen too high and the debt to the state will be paid with at least probability $p\cdot 100\%$.
\\Assume the percentage $p$ is fixed contractually, then we can find the optimal pair $(b,\alpha)$ such that the expected loss of the PC is minimised. Let
\begin{align*}
V(b):&=\mE\big[R_1(b)\big]=\mathbb{E}\Bigg[ e^{\mu + \sigma W_{1} -  \left(\max \limits_{0 \leq s \leq 1}  \{\mu s+\sigma W_{s}\} -\ln(1+b) \right)^{+}}\Bigg]\;,
\\U(b):&=\mE[D_1(b)]=(1+b)\mE\Big[\Big(\max\limits_{0\le s\le 1}\{\mu s+\sigma W_s\}-\ln(1+b)\Big)^+\Big]\;,
\end{align*}
i.e. $V(b)$ is the expected return on investment of the PC and $U(b)$ is the expected value of the accumulated debt account.
\\In order to proceed with our derivations and also for numerical calculations we will need to consider the density function of $\max\limits_{0\le s\le t}\{\mu s+ \sigma W_s\}$. This density, $g(y;t)$, is given by
\begin{align}
g(y;t,\mu)=\frac2{\sqrt{2\pi t}\sigma}e^{-\frac{(y-\mu t)^2}{2\sigma^2 t}}-\frac {e^{\frac{2\mu y}{\sigma^2}}\mu}{\sigma^2} {\rm Erfc}\Big(\frac{y+\mu t}{\sqrt{2 t}\sigma}\Big),\label{density}
\end{align}
with ${\rm Erfc(x)}=\frac{2}{\sqrt{\pi}}\int_x^\infty e^{-z^2}\md z$, confer for instance Borodin \& Salminen, p. 250 formula 1.1.4 for the distribution function of  $\max \limits_{0 \leq s \leq 1}  \{\mu s+\sigma W_{s}\}$. 
\\For simplicity we write $g(y;\mu)$ if $t=1$. \medskip
\\The normalised loss (the spent amount exceeding the required payments of $C_1-C_0$, i.e. the real loss divided by $C_1-C_0$) is given by
\begin{align}
L(b,\alpha):=\alpha  -\alpha V(b)-1\to\min!\label{optim}
\end{align}
It is clear that $V(b)$ is strictly increasing in $b$, meaning that the maximum of $V(b)$, and the minimum of the loss, is attained at the maximal $b$, allowed by the credibility condition \eqref{credibility}. However, this maximal $b$ will again depend on $\alpha$. Thus, we have to specify the set of admissible pairs $(b,\alpha)$ before we can solve the optimisation problem \eqref{optim}. Note first that the PC does not have an infinite amount of money on her/his disposal. Therefore, we have to restrict the set of admissible $\alpha$ to some reasonable values and introduce a liquidity restriction boundary $\alpha^*>0$ with $0<\alpha\le \alpha^*$. 
\begin{Lem}\label{Bem:1}
Assume $p$ in \eqref{credibility} is fixed, $\tilde p$ is given by
\[
\int_{\tilde p}^\infty g(y;t,\mu)\md y = p
\] 
and $\alpha^*<\infty$ is the liquidity restriction boundary. 
\\If $\alpha^*< e^{1-\tilde p}$ then the PC should prefer to pay the required increase in contribution directly into the PAYG account. 
\\If $\alpha^*\ge e^{1-\tilde p}$, the pair minimising the loss \eqref{optim} is $(b^*,\alpha^*)$, where $b^*$ is implicitly given by
\[
\frac 1{(1+b^*)(\tilde p-\ln(1+b^*))}=\alpha^*
\]
\end{Lem}
\begin{proof}
Consider the credibility condition \eqref{credibility}. Using the density \eqref{density}, we get
\begin{align*}
\mP\Big[D_t(b)\ge \frac 1\alpha\Big]&=\mP\Big[(1+b)\Big(\max\limits_{0\le s\le t}\{\mu s+\sigma W_s\}-\ln(1+b)\Big)^+\ge \frac 1\alpha\Big]
\\&= \mP\Big[\max\limits_{0\le s\le t}\{\mu s+\sigma W_s\}\ge \frac 1{(1+b)\alpha}+\ln(1+b)\Big]
\\&=\int_{\ln(1+b)+\frac1{\alpha(1+b)}}^\infty g(y;t,\mu)\md y\;.
\end{align*}
Since $g(y;t,\mu)>0$ for all $y\in\R_+$, for every $p\in[0,1]$ there is a unique lower integral boundary $\tilde p\in\R_+$ such that
\[
\int_{\tilde p}^\infty g(y;t,\mu)\md y = p\;.
\]
Assume $p$ in \eqref{credibility} is fixed and $\tilde p$ is the corresponding lower integral boundary. Then the set of admissible pairs $(b,\alpha)$ is described by the inequality 
\[
\tilde p\ge \ln(1+b)+\frac1{\alpha (1+b)}\;,
\]
which is equivalent to $\alpha \ge \frac 1{(1+b)(\tilde p-\ln(1+b))}=:\Delta(b)$. Further, it holds 
\[
\Delta'(b)\begin{cases}\le 0\mbox{ $:$ $e^{\tilde p-1}-1\ge b$},
\\>0 \mbox{ $:$ $e^{\tilde p-1}-1< b$},
\end{cases}
\] 
i.e. $\Delta(b)$ attains its global minimum at $b=e^{\tilde p-1}-1>-1$ with $\Delta(e^{\tilde p-1}-1)=e^{-\tilde p+1}$. 
\\Since our target is to minimise the loss defined in \eqref{optim}, and it is done by the biggest allowed value of $b$, we can restrict our considerations to the set $b\ge e^{\tilde p-1}-1$, the area where $\Delta(b)$ is strictly increasing.  
\\Noting that $e^{\tilde p}-1>e^{\tilde p-1}-1$ and $\Delta(e^{\tilde p}-1)=\infty$, we define the set of allowed pairs $(b,\alpha)$ to be
\[
e^{\tilde p-1}-1\le b< e^{\tilde p}-1\;\; \mbox{and} \;\; \alpha\ge \frac 1{(1+b)(\tilde p-\ln(1+b))}\;.
\]
 Note that if $\alpha^*< \Delta(e^{\tilde p-1}-1)=e^{-\tilde p+1}$, then the set of admissible pairs is empty. 
\\If $\alpha^*\ge e^{-\tilde p+1}$, there is a unique $b^*\in [e^{\tilde p-1}-1,e^{\tilde p}-1)$ such that $\Delta(b^*)=\alpha^*$. And the set of admissible pairs shrinks to
\[
e^{\tilde p-1}-1\le b< b^*\;\; \mbox{and} \;\;  \frac 1{(1+b)(\tilde p-\ln(1+b))}\le\alpha\le\alpha^*\;.
\]
\end{proof}
The above lemma ensures that if we require $\alpha^*\ge e^{1-\tilde p}$, the minimal loss will be attained at $(b^*,\alpha^*)$. If $L(b^*,\alpha^*)$ is bigger than zero, this is a clear indicator that the funding strategy is not working, and paying the increase in contribution immediately into the PAYG account is more preferable for the PC.
\bigskip
\\
Note that if the accumulated amount exceeds the debt then the PC has an additional gain. Therefore, the entire expected normalised loss function $L_e$ is defined as follows:
\[
L_e(b,\alpha):=\alpha-\alpha V(b)-1-\alpha U(b)+1= \alpha-\alpha V(b)-\alpha U(b)\;.
\]
\begin{Lem}
The function $L_e$ is decreasing in $b$, if $\alpha>0$.
\end{Lem}
\begin{proof}
W.l.o.g we assume $\alpha=1$.
Using \eqref{density}, we get the following representation
\begin{align*}
V(b)+U(b)&=e^{\mu+\frac{\sigma^2}2}\int_0^{\ln(1+b)} g(y;\mu+\sigma^2)\md y
\\&{}+(1+b)e^{\mu+\frac{\sigma^2}2}\int_{\ln(1+b)}^\infty e^{-y}g(y;\mu+\sigma^2)\md y
\\&{}+(1+b)\int_{\ln(1+b)}^\infty (y-\ln(1+b))g(y;\mu)\md y\;.
\end{align*}
Deriving with respect to $b$ and using $e^z-1\ge z$ yields
\begin{align*}
V'(b)+U'(b)&=
\mE\Big[e^{\mu+\sigma W_1-M_1}\one_{[M_1>\ln(1+b)]}\Big]
\\&\quad {}+ \mE\Big[\Big(M_1-\ln(1+b)-1\Big)\one_{[M_1>\ln(1+b)]}\Big]
\\&\ge \mE\Big[\Big(\mu+\sigma W_1-\ln(1+b)\Big)\one_{[M_1>\ln(1+b)]}\Big]\;.
\end{align*}
In order to calculate the expectation in the last line above, we use the Markovian property of the Brownian motion $W$. Define additionally $\tau_b:=\inf\{t\ge 0: \mu t+\sigma W_t=\ln(1+b)\}$, then
\begin{align*}
 \mE&\Big[\Big(\mu+\sigma W_1-\ln(1+b)\Big)\one_{[M_1>\ln(1+b)]}\Big]
 \\&= \mE\Big[\Big(\mu \tau+\sigma W_\tau-\ln(1+b)+\mu(1-\tau)+\sigma\big(W_1-W_\tau\big)\Big)\one_{[\tau<1]}\Big]
 \\&= \mE\Big[\Big(\mu(1-\tau)+\sigma  \tilde W_{1-\tau}\Big)\one_{[\tau<1]}\Big]=\mu\mE[(1-\tau)\one_{[\tau<1]}]> 0\;,
\end{align*}
where $\tilde W$ is an independent copy of $W$. 
\end{proof}
This means in particular that $L_e$ attains its minimum at the biggest admissible $b$. Of course, one might minimise the function $L_e$ instead of $L$ and redefine the set of admissible barriers in \eqref{credibility} accordingly. However, we keep the credibility condition \eqref{credibility} and consider the additional (positive) expectation $\alpha U(b)-1$ as a buffer similar to the net profit condition in risk theory, confer for instance \cite[p.\ 130]{Dickson}, where the expected premia should be strictly bigger than the expected loss in order to avoid the almost sure ruin. In our model, this means that we keep $\alpha\mE[D_1(b)]> C_1-C_0$ in order to make the repayment of the debt more probable even if the probability $p$ in \eqref{credibility} is relatively small. \smallskip
\\Since it does not make sense for the PC to prefer the funding scheme to paying the increase in contribution directly if the expected loss is positive, we formalise the following (normalised by dividing through $C_1-C_0$) profitability condition
\begin{equation}
\alpha -\alpha V(b)=\alpha -\alpha \mathbb{E} \left[ e^{\mu+\sigma W_1-\big(\max \limits_{0 \leq s \leq 1}  \{\mu s+\sigma W_{s}\} -\ln(1+b)\big)^+}  \right] < 1 \;. \label{profit}
\end{equation}
The above condition means that the normalised loss of the PC $\alpha -\alpha V(b)$ should be smaller than the amount of $1$ that could have been paid directly into the PAYG. Recall that in order to create a buffer for the state in the sense that the debt will be fully repaid in expectation, we do not take into account the possible gain if the debt account exceeds the $C_1-C_0$.
\\In order to rewrite the above conditions in terms of integrals, we introduce a new measure $Q$ by the following Radon-Nicodym density $\frac{\md Q}{\md \mP}=e^{\sigma W_1-\frac{\sigma^2}2}$ and $\tilde W_s=W_s-\sigma s$ a standard Brownian motion under $Q$. Setting for simplicity $M_1:=\max \limits_{0 \leq s \leq 1}  \{\mu s+\sigma W_{s}\}$, we can rewrite the expected return on investment for the PC as follows using the density introduced in \eqref{density}
\begin{align*}
&V(b)=\mE\Big[e^{\mu + \sigma W_1} e^{-\left(M_1-\ln(1+b)\right)^+}\Big] =e^{\mu+\frac{\sigma^2}2}\mE_Q\Big[e^{-\left(M_1-\ln(1+b)\right)^+}\Big]\nonumber
\\&=e^{\mu+\frac{\sigma^2}2}\Big\{\int_0^{\ln(1+b)}g(y;\mu+\sigma^2)\md y+(1+b)\int_{\ln(1+b)}^\infty e^{-y}g(y;\mu+\sigma^2)\md y\Big\}\;.
\end{align*}
\begin{Bem}
Consider the derivatives of $V(b)$:
\begin{align*}
&V'(b)=e^{\mu+\frac{\sigma^2}2}\int_{\ln(1+b)}^\infty e^{-y}g(y;\mu+\sigma^2)\md y
\\&\quad\quad \,=\mE\Big[e^{\mu+\sigma W_1-\max \limits_{0\le s\le 1}\{\mu s +\sigma W_s\}}\one_{[\max \limits_{0\le s\le 1}\{\mu s +\sigma W_s\}\ge \ln(1+b)]}\Big]>0,
\\&V''(b)=-\frac{e^{\mu+\frac{\sigma^2}2}}{(1+b)^2}\cdot g(\ln(1+b);\mu+\sigma^2)<0\;.
\end{align*}
By maximising just the expected value, we do not take into account the variance and consequently the risk for return on investment to stay considerably below the expected value. One might consider the following target functional instead of $V(b)$.
\[
\tilde V(b):=\mE\Big[e^{\mu + \sigma W_1-\left(M_1-\ln(1+b)\right)^+}\Big]-\lambda \mE\Big[e^{2\mu + 2\sigma W_1-2\left(M_1-\ln(1+b)\right)^+}\Big]\to\max!
\]
with some weight $\lambda>0$. Using change of measure technique mentioned above, the derivative fulfils
\begin{align*}
\tilde V'(b)&= \mE\Big[\Big\{e^{\mu+\sigma W_1-M_1}-2\lambda(1+b)e^{2\mu+2\sigma W_1-2M_1}\Big\}\one_{[M_1>\ln(1+b)]}\Big]
\\&=e^{\mu+\frac{\sigma^2}2}\int_{\ln(1+b)}^\infty e^{-y} g(y;\mu+\sigma^2)\md y
\\&\quad {}-2\lambda(1+b)e^{2\mu+2\sigma^2}\int_{\ln(1+b)}^\infty e^{-2y} g(y;\mu+2\sigma^2)\md y
\end{align*}
Since $\mu+\sigma W_1\le M_1$ a.s. we can conclude that $e^{\mu+\sigma W_1-M_1}\ge e^{2\mu+2\sigma W_1-2M_1}$ a.s., meaning that for $2\lambda (1+b)\le 1$ the first derivative $\tilde V'(b)$ stays positive, meaning that the value $\tilde V(b)$ is increasing. However, the global behaviour of $\tilde V'(b)$ is highly sensitive to the choice of $\lambda$ - a variable which cannot be clearly justified from the economical point of view and creates in this way another source of uncertainty and an opportunity for manipulations. 
\\
Also, one should not forget that the derivative $\tilde V'(b)$ can not be considered globally, but just on the interval allowed by the credibility condition \eqref{credibility} and the boundary $\alpha^*$ defined in Remark \eqref{Bem:1}. In Figure \ref{profit1} we see the functions $V(b)$ (left) and $\tilde V(b)$ (right) for $\mu=0.04$, $\sigma=0.2$ and $\lambda=0.85$. Depending on the $b^*$ the maximum of $\tilde V(b)$ will be attained at different values of $b$. Also, the curve $\tilde V(b)$ changes its form depending on $\lambda$.
\begin{figure}[t]
\includegraphics[scale=0.45, bb = 70 340 200 760]{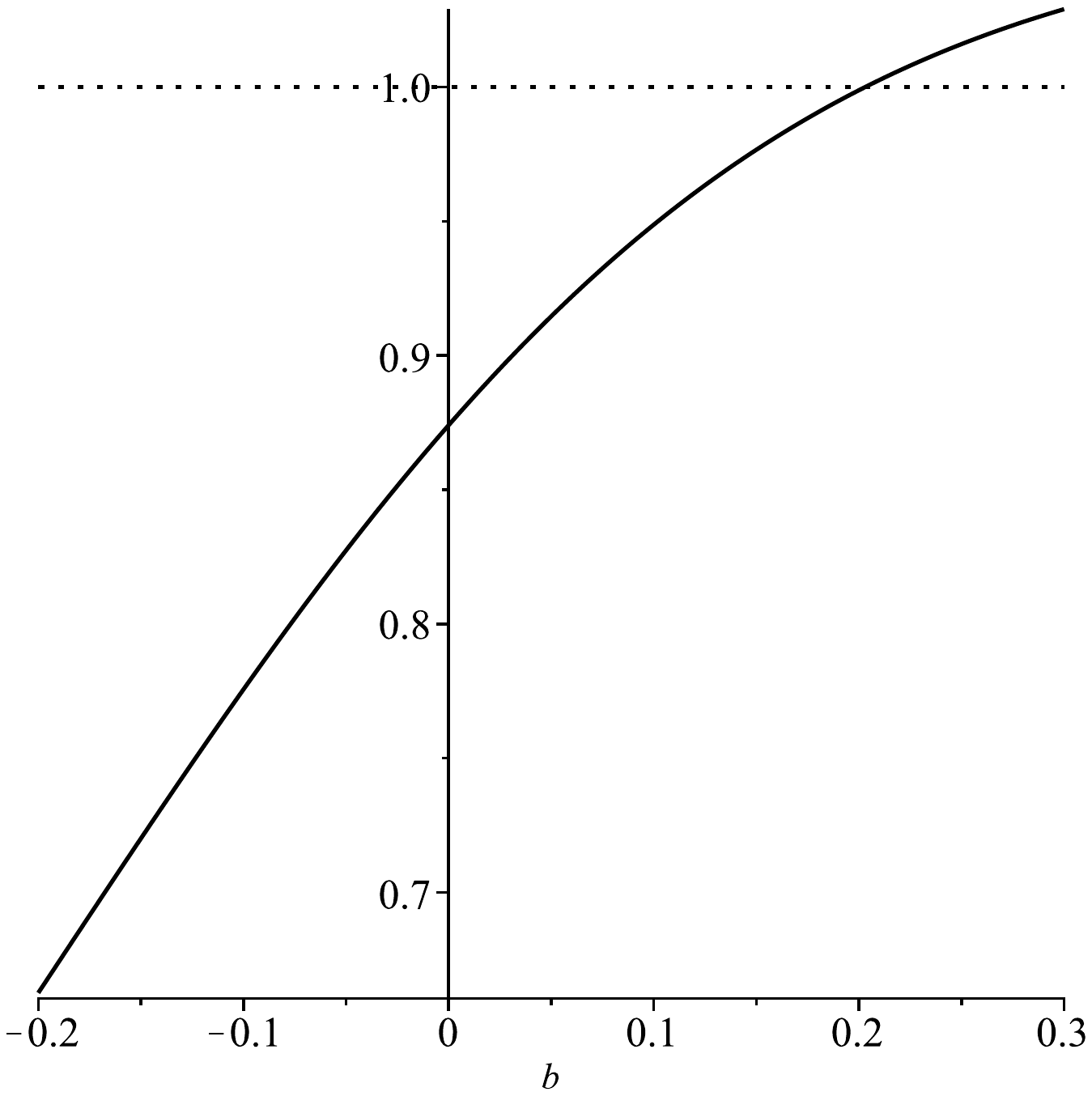}
\includegraphics[scale=0.45, bb = -200 340 200 760]{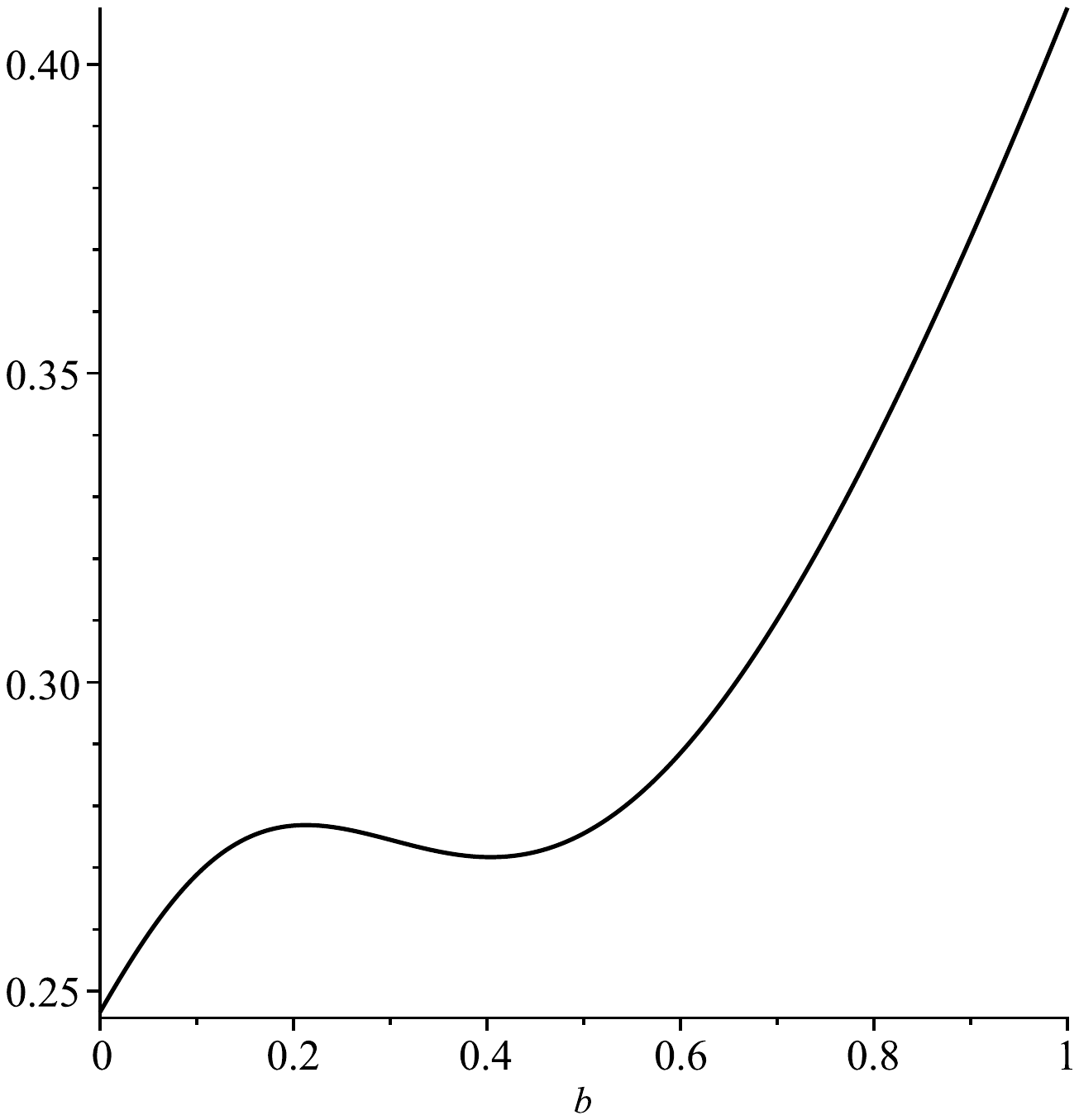}
\caption{Expected \textbf{1-year} return on investment $V(b)$ and modified functional $\tilde V(b)$ with $\lambda=0.85$. \label{profit1}}
\end{figure}
\end{Bem}
In the examples below we demonstrate how the credibility and the profitability conditions work for a fund with realistic parameters.   
\begin{Bsp}\label{Bsp1}
Assume like above $\mu=0.04$, $\sigma=0.2$ and $\alpha^*=10$.
\\In the table below, we compute several values of $\mP\big[D_1(b)\ge \frac 1\alpha\big]$ for different pairs $(b,\alpha)$.
\begin{table}[H]
\centering
		\begin{tabular}{r|ccccc}
		  $\alpha $  & $1$ & $2$& $3$ & $4$ & $5$ \\
 \mbox{$b$}\hspace{0.6cm}   &       &       &    &    &\\\hline
$- 0.2$ &$7.7\cdot 10^{-7}$ 	& 0.06553328 & 0.40025361   & 0.71208305   & 0.91564864\\		
$-0.1$  &$1.32 \cdot 10^{-6}$ & 0.03766247  & 0.23837807   & 0.45738957   & 0.62378926\\
$-0.05$ &$1.48\cdot 10^{-6}$ & 0.02776680  & 0.17865039   & 0.35389069   & 0.49496449\\
$0$     &$1.53\cdot 10^{-6}$ & 0.02014832  & 0.13156264   & 0.26808831   & 0.38351917\\
$0.05$  &$1.49\cdot 10^{-6}$ & 0.01441359  & 0.09534295   & 0.19919441   & 0.29075048 \\
		\end{tabular}	
		\caption{$\mP\big[D_1(b)\ge \frac 1\alpha\big]$ for different values of $\alpha$ and $b$ for 1 year time horizon.}
\end{table}
\bigskip
\noindent
We see that if the state requires a relatively high $p$, the values of $\alpha$ are high and $b$ becomes negative. \smallskip
\\$\bullet$ Choose now $p=0.7$, i.e. the credibility condition is $\mP\big[D_1(b)\ge \frac 1\alpha\big]\ge 0.7$. Then: $\tilde p=0.093078333$, $\alpha^*$ should be bigger than $2.4766867$, the biggest possible $b$ leading to $\alpha=\infty$ equals $0.0975477$ and $b^*=-0.00768$. It holds
\[
L(b^*,\alpha^*)= \alpha^* -\alpha^* V(b^*)-1=0.327634>0\;.
\]
Thus, if the state requires $p=0.7$, the PC should prefer to pay the increase in contribution $C_1-C_0$ into the PAYG her-/himself and not to use the funding possibility.
\medskip
\\$\bullet$ If the state requires $\mP\big[D_1(b)\ge \frac 1\alpha\big]\ge 0.5$ then $\tilde p=0.15750112$. The minimal possible $\alpha^*$ preventing the set of admissible pairs $(b,\alpha)$ to be non-empty equals $e^{-\tilde p+1}=2.3221625$. The biggest possible $b$ leading to $\alpha=\infty$ equals $0.1705821$ and the maximal admissible $b$ corresponding to $\alpha^*=10$ is given by $b^*=0.06574$. The expected loss is then
\[
L(b^*,\alpha^*)=-0.2603<0\;.
\]
\begin{figure}[t]
\includegraphics[scale=0.45, bb = 70 310 200 650]{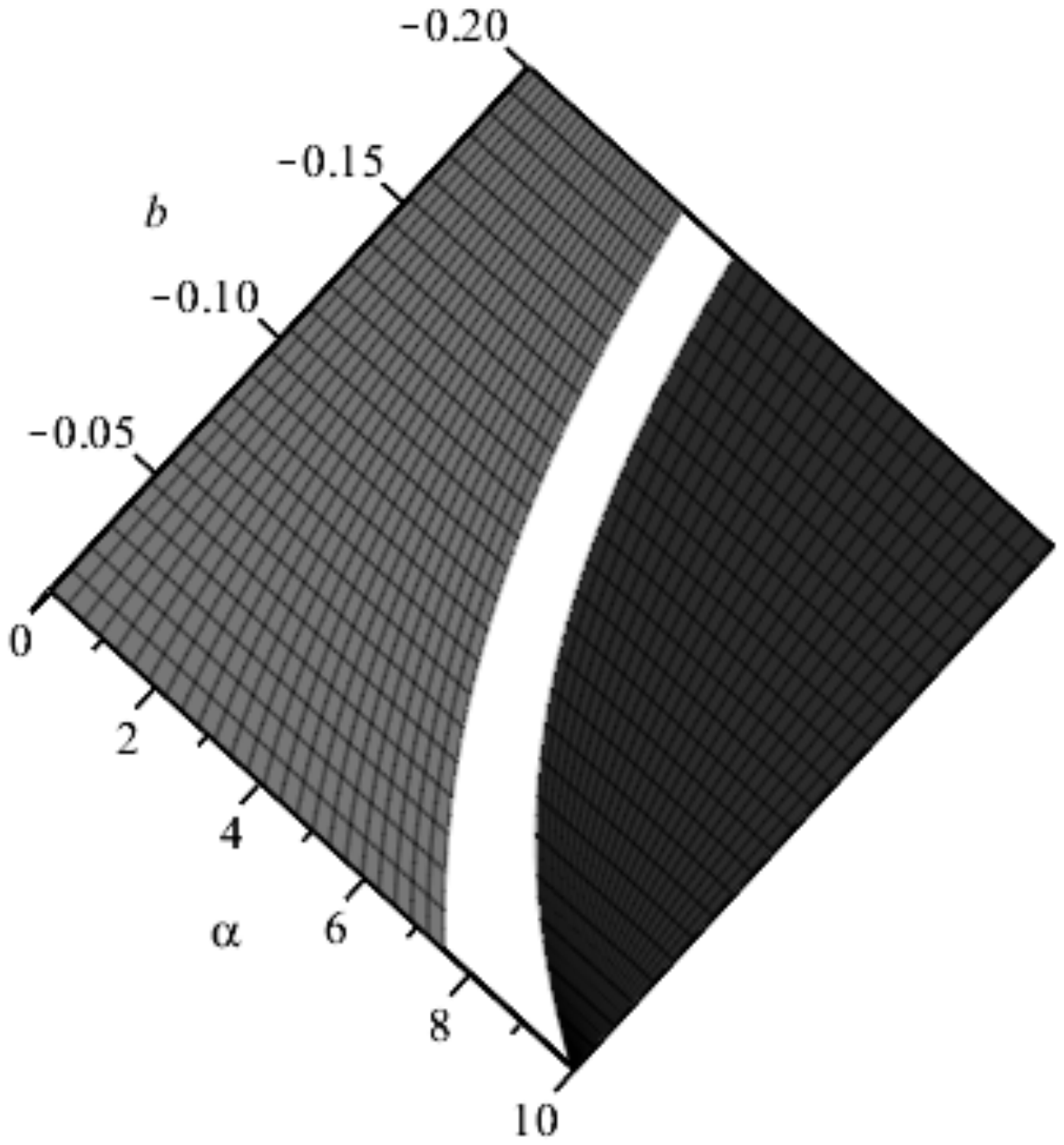}
\includegraphics[scale=0.45, bb = -200 310 200 650]{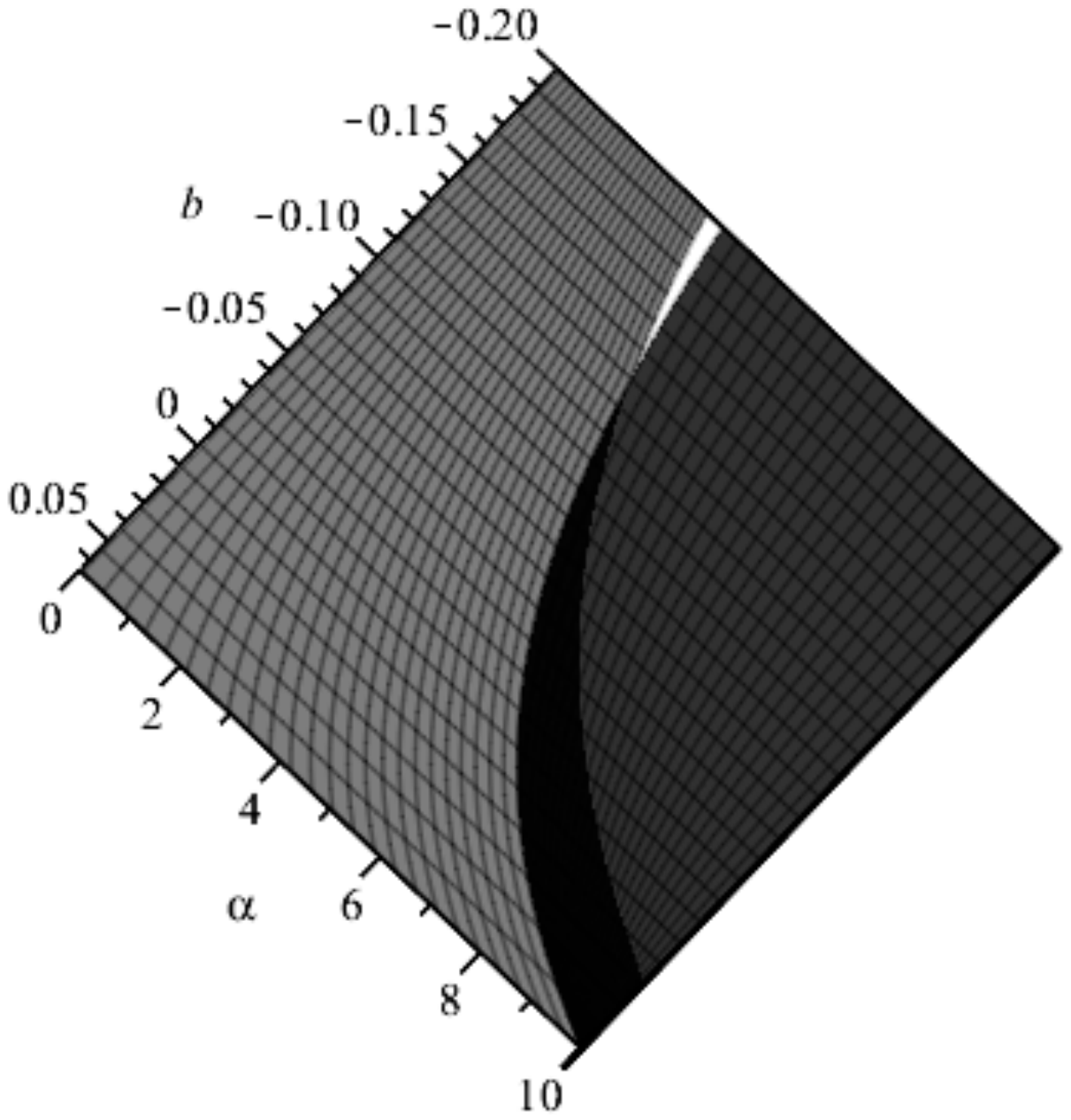}
\caption{Sets of $(b,\alpha)$ fulfilling the credibility \eqref{credibility} (dark grey areas) and profitability \eqref{profit} (light grey areas) conditions for $p=0.7$ (left) and $p=0.5$ (right). \label{intersection}}
\end{figure}
What does this result mean for the PC? The function $V(b)$ (expected return on investment for the PC) is given in Figure \ref{profit1}, left picture. If $b<0.2030$ then the expected return on investment will be smaller than $100\%$, meaning that in expectation the PC will not get her/his full investments back. However, the primary target of adding a funding component is to reduce the increase in contribution $C_1-C_0$ and not to entirely avoid any payments. Below we demonstrate how this might work in practice.
\\In Germany, it is planned to increase the monthly contribution from $18.6\%$ to $19.5\%$ starting from the year 2024. Assume, for an average contributor the monthly increase will approximately amount to $20$ Euro, i.e. $240$ Euro per year. Thus, investing $240\cdot\alpha^*=2,400$ Euro into a fund with parameters given like above, the PC will pay the debt of $240$ Euro with probability $50\%$ and, compared to the direct payment to the PAYG, have a gain of $240\cdot 0.2603\approx 62.5$ Euro. This means, the PC will pay in expectation $240-62.5=177.5$ instead of $240$ Euro per year. The expected loss of the state is then given by $$(C_1-C_0)-(C_1-C_0)\alpha^*\cdot U(b^*)=-75.6,$$ meaning that in expectation the PC gets back an additional amount of $75.6$. Thus, the PC has an expected gain of $62.5+75.6=138.1$ compared to paying the increase of contribution immediately. And the state gets the debt in expectation fully back. 
\bigskip  
\\In Figure \ref{intersection} we plotted the sets of $(b,\alpha)$ fulfilling the credibility \eqref{credibility} and the profitability \eqref{profit} conditions.
\\On the left picture of Figure \ref{intersection} we see the area with $(b,\alpha)$ such that $V(b,\alpha)\ge 0.7$ (dark grey) and the area where $L(b,\alpha)\le 0$ (light grey). It is obvious that these two sets are disjoint, meaning that there is no combination of $(b,\alpha)$ for $\alpha^*=10$ such that the debt is repaid with at least probability $70\%$ and simultaneously the PC's loss is less than $C_1-C_0$. In the right picture in Figure \ref{intersection}, we see that the set of pairs $(b,\alpha)$ fulfilling the credibility condition with $p=0.5$ (dark grey) and simultaneously the profitability condition (light grey) is given by the narrow black area lying between the dark and light grey ones. 
\end{Bsp} 
\subsection{A $10$-years investment period}
It is clear that investing for a longer period brings higher returns in expectation.  Therefore, in this section, the state still pays the increase in contributions, but requires its money back after a period of 10 years. The PC invests some amount of money for 10 years in a pre-specified fund. At the end of the investment period, the PC has to repay the debt and will hopefully get some gain.
\\The invested amount should be a multiple, $\alpha$, from the total anticipated increase of contributions over the next ten years. Considering the numbers given in Example \ref{Bsp1}, the forecast of increase in contribution over the next 10 years amounts to $2,400$ Euro. 
Like in the previous section, we require the credibility condition \eqref{credibility} and the profitability condition \eqref{profit} while the time interval changed from $1$ to $10$.
\begin{Bsp}
\begin{figure}[t]
\includegraphics[scale=0.45, bb = -150 350 200 650]{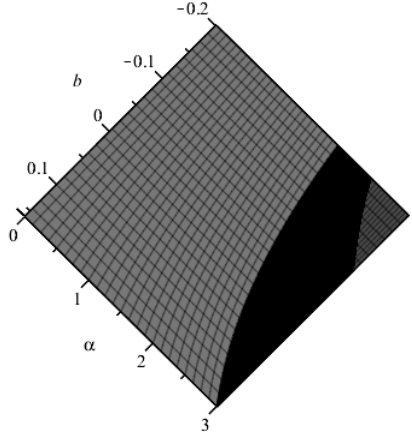}
\caption{Sets of $(b,\alpha)$ fulfilling the credibility \eqref{credibility} (dark grey area) and profitability \eqref{profit} (light grey area) conditions for \textbf{10 years} and $p=0.7$.\label{bild3}}
\end{figure}
Assume again $\mu=0.04$ and $\sigma=0.2$.  
In the Table \ref{tab10} below we calculate $\mP[D_{10}(b)\ge \frac1\alpha]$ for different pars of $(b,\alpha)$. The difference to a 1 year investment is considerable. For instance, in order to keep the probability of repayment above $70\%$ it suffices to set $\alpha=3$ and $b=0.05$. Therefore, we set $\alpha^*=3$, i.e. the highest possible $\alpha$ the state is willing to adopt is equal to $3$.                        
\begin{table}[H]
	\centering
		\begin{tabular}{r|ccccc}
		 $\alpha $  & $1$ & $2$& $3$ & $4$ & $5$ \\
			 \mbox{$b$}\hspace{0.6cm}   &       &       &    &    &\\\hline
$-0.2$ & 0.2546296 & 0.7276311 & 0.8842143   & 0.9508487   & 0.9860533     \\
$-0.1$ & 0.2671423 & 0.6884180 & 0.8334411   & 0.8984640   & 0.9341821\\
$-0.05$& 0.2698034 & 0.6680602 & 0.8076430   & 0.8715453   & 0.9071802\\
$0$    & 0.2706137 & 0.6474819 & 0.7817909   & 0.8444029   & 0.8797703 \\
$0.05$ & 0.2699046 & 0.6268430 & 0.7560225   & 0.8172035   & 0.8521476 \\
		\end{tabular}
		\caption{$\mP\Big[D_{10}(b)\ge \frac 1\alpha\Big]$ for different values of $\alpha$ and $b$ for 10 years time horizon.\label{tab10}}                       
\end{table}
\noindent
Figure \ref{bild3} showcases the pairs $(b,\alpha)$ allowed by the credibility condition \eqref{credibility} (dark grey area), the pairs $(b,\alpha)$ where the profitability condition is fulfilled (light grey area) and the intersection area (black). It is clear that the 10-year investment should be preferred if the financial situation allows.  
\end{Bsp}
\section{Comparison of the Two Investment Types}
This section compares the two strategies introduced above in Sections \ref{sec:ls} and \ref{sec:opti_approach}: continuous withdrawal and a lump sum repayment at the end of the period. 
\\For that purpose, we look at the loss functions corresponding to each strategy in dependence on time. Let again $\alpha^*$ be the liquidity restriction. Define 
\begin{align*}
&L_d(t):=\alpha^* -\alpha^*\mE[e^{\mu t+\sigma W_t}]+1,
\\&L_c(t,b):= \alpha^* -\alpha^*\mE\big[e^{\mu t+\sigma W_t-\big(\max\limits_{0\le s\le t}\{\mu s+\sigma W_s\}-\ln(1+b)\big)^+}\big]-1\;.
\end{align*}
i.e. $L_d$ is the normalised loss of the PC in the model with the lump sum debt repayment and $L_c$ is the normalised loss in the continuous withdrawal model. Both functions depend now on the time interval under consideration. Let further
\begin{align*}
\Lambda(t,b)&:=L_c(t,b)-L_d(t)
\\&=\alpha^*\mE[e^{\mu t+\sigma W_t}]-\alpha^*\mE\big[e^{\mu t+\sigma W_t-\big(\max\limits_{0\le s\le t}\{\mu s+\sigma W_s\}-\ln(1+b)\big)^+}\big]-2\;.
\end{align*}
It is easy to see that $\Lambda$ is strictly increasing in $t$ and strictly decreasing in $b$. 
Because 
\[
\Lambda(0,b)=\begin{cases}
-2&\mbox{: $b\ge 0$}\\
-\alpha^*b-2&\mbox{: $b< 0$},
\end{cases}
\]
and $\lim\limits_{t\to\infty} \Lambda(t,b)=\infty$ we conclude that for every $b$ with $b\ge-\frac 2{\alpha^*}$ there is an $t$ such that $L(t,b)=0$. For $b< -\frac 2{\alpha^*}$ the function $\Lambda$ stays positive meaning that the continuous withdrawal strategy is definitely not optimal.
\\Thus, for $b\ge-\frac 2{\alpha^*}$ we can conclude that by implicit function theorem there is a curve $\beta:[0,\infty)\to[-\frac 2{\alpha^*}\vee -1,1]$, $\beta'>0$ such that $\Lambda(t,\beta(t))\equiv 0$, $\Lambda(t,b)>0$ for $b<\beta(t)$ and $\Lambda(t,b)<0$ for $b>\beta(t)$. 
\\The choice of the investment strategy will depend on the parameters of the underlying fund, the liquidity restriction $\alpha^*$, the investment time horizon $t$ and credibility condition \eqref{credibility}:
\begin{itemize}
\item If $\Lambda(t,b^*)>0$, choose the lump sum repayment strategy;
\item If $\Lambda(t,b^*)\le 0$ choose the continuous withdrawal strategy;
\item If $L_d(t)$,$L_c(t,b^*)\ge 0$ - pay directly into the PAYG system.
\end{itemize}  
Assuming again $\mu=0.04$, $\sigma=0.2$, $\alpha^*=10$, $t=1$ and $p=0.5$ we get $b^*=0.06574$, confer Example \ref{Bsp1}. It holds that
\begin{align*}
& L_d(1)=10-10e^{\mu+\frac{\sigma^2}2}+1=0.3816>0,
\\& L_c(1,b^*)=-0.2603<0\;.
\end{align*}
Therefore, for one year time horizon and $\alpha^*=10$ one should prefer the continuous withdrawal. 
\\For the time horizon of 10 years one gets for the same parameters $b^*=0.8707$
\begin{align*}
&L_d(10)=10-10e^{(\mu+\frac{\sigma^2}2)10}+1=-7.2211<0,
\\&L_c(10,b^*)=-3.5291<0,
\\& \Lambda(10,b^*)=3.6921>0,
\end{align*}
meaning that one should definitely prefer the lump sum repayment strategy. This can be explained by the fact that by withdrawing money from the investment we miss possible gains. Even taking into account the amount on the debt account exceeding the actual debt will not make the continuous withdrawal more attractive: 
\[
L_c(10,b^*)-\alpha^*\mE[D_{10}(b^*)]+1=-7.088>-7.2211\;.
\]
In the table below we sum up the optimal choice of a strategy for different values of $\alpha$ and time horizons $t$ in years. Let C denote the continuous withdrawal and LS the lump sum debt repayment strategy. If it is optimal not to use any of the funding strategies but just to pay the increase in contribution into the PAYG, we write PAYG.
\\
\begin{table}[h]
\centering
		\begin{tabular}{r|cccccccccc}
$\alpha$ & $1$ & $2$ & $3$  & $4$  & $5$  &  $6$ & $7$ &$8$    & $9$  &$10$ \\
 \mbox{$t$}\hspace{0.6cm} & & & & & & & &\\\hline
$1$ &  PAYG & PAYG & PAYG & C & C & C & C & C & C & C\\  
$2$ &  PAYG & PAYG & C  & C & C & C & C & C & C & C\\		
$4$ &  PAYG & C  & C  & C & C & C & C & LS & LS & LS\\
$6$ &  PAYG & C  & C  & C & LS & LS & LS & LS & LS & LS\\
$8$ &  PAYG & C  & C  & LS & LS & LS & LS & LS & LS & LS \\
$10$&  PAYG & C  & LS & LS & LS & LS & LS & LS & LS & LS \\
$20$&  LS   & LS & LS & LS & LS & LS & LS & LS & LS & LS
\\
$40$&  LS   & LS & LS & LS & LS & LS & LS & LS & LS & LS
\\

		\end{tabular}	
		\caption{The optimal strategy for different values of $\alpha$ and time horizons $t$.}
\end{table}
\\
The probability of default for the state, i.e. the probability that the state will not get the full debt back, equals $50\%$ by definition of the optimal $b^*$. For the discrete repayment case we sum up the results in Table \eqref{default} below for the parameters given above. For each pair $(t,\alpha)$ we calculate the probability of default $\mP[\alpha e^{\mu t+\sigma W_t}<1+\alpha]$, meaning that the state will get less than $C_1-C_0$ after the PC gets his invested money back.  
\begin{table}[H]
\centering
		\begin{tabular}{r|cccccccccc}
$\alpha$ &$1$ & $2$ & $3$  & $4$  & $5$  &  $6$ & $7$ &$8$    & $9$  &$10$ \\
 \mbox{$t$}\hspace{0.6cm} & & & & & & & &\\\hline
$1$ &0.99 &0.96 &0.89 & 0.82 & 0.76 & 0.72 & 0.68 & 0.65 & 0.63 & 0.61 \\

$2$ &0.98 &0.87 &0.77 & 0.69 & 0.64 & 0.60 & 0.58 & 0.55 & 0.54 & 0.52\\		
$4$ &0.91 &0.73 &0.63 & 0.56 & 0.52 & 0.49 & 0.47 & 0.46 & 0.45 & 0.44\\
$6$ &0.82 &0.63 &0.54 & 0.49 & 0.45 & 0.43 & 0.41 & 0.40 & 0.39 & 0.38\\
$8$ &0.75 &0.56 &0.48 & 0.43 & 0.40 & 0.38 & 0.37 & 0.36 & 0.35 & 0.34 \\
$10$&0.68 &0.50 &0.43 & 0.39 & 0.37 & 0.35 & 0.34 & 0.33 & 0.32 & 0.31 \\
$20$&0.45 &0.33 &0.28 & 0.26 & 0.25 & 0.24 & 0.23 & 0.22 & 0.21 & 0.21\\
$40$&0.24 &0.17 &0.15 & 0.14 & 0.13 & 0.13 & 0.12 & 0.12 & 0.12 & 0.12 \\
		\end{tabular}	
		\caption{$\mP[\alpha e^{\mu t+\sigma W_t}<1+\alpha]$ for different values of $\alpha$ and time horizons $t$.\label{default}}
\end{table}
\noindent Again, in Table \ref{default} we see that for short-term periods it is more profitable to use the continuous withdrawal strategy or PAYG.
\section{Conclusions}
In the last decade, most OECD countries have enacted pension reforms of their traditional defined-benefit pay-as-you-go (PAYG) schemes. PAYG requires the balance between income from contributions and pension expenditure where the current contributions finance the current pensions. The most common reforms have been the changes in the level of benefits (sometimes linked to a longevity index, such as the life expectancy) and increases in the retirement age. \\
At the same time, there are some countries that combine the PAYG scheme and a defined contribution funding part within the mandatory pension system. These systems have been advocated, particularly by the World Bank, as a practical way to the higher financial market returns with the cost of a scheme with a greater funded component. In this line, the present paper proposes an alternative with the contributors investing an extra amount of money into a fund, so that part of the investment together with the returns can restore the financial sustainability of the PAYG scheme. The contributor faces the investment risk while the state bears the risk that the returns on investment do not cover the deficit in contributions for the period analysed \color{black} However, in the case of extra returns after debt repayment, the contributor keeps the gains. On the other hand, the loss of the state, in the worst case, will amount to the increase in contribution needed to restore the financial sustainability of the system while the contributor would normally invest and risk a much bigger multiple of this amount.\\
Two different debt repayment types depending on the amount invested and the timing of the repayment to the state are analysed for a prototypical contributor. The first set of strategies features the repayment of the debt at the end of a pre-specified period as a lump sum. In several examples, we calculate the probability of full payback of the debt, expected loss of the state and the expected value of the gains for the individual. In particular, we compare, for different invested amounts, the case of the repayment after a year versus a repayment after 10 years. As expected, the loss probability decreases when the amount invested and the investment horizon increase. \\
As an optimisation approach, we study the case when the state requires the individual to transfer any excess above a particular level (barrier) of return – to be optimised -- continuously in time as a debt repayment (second type of strategies). Comparing different barriers, we show that the optimal strategy is the biggest possible barrier such that the debt to the state is repaid with a certain pre-specified probability. In an example, we compare the continuous withdrawal and the lump sum repayment strategies with the possibility to pay the increase in contribution directly into the PAYG system. If the period for the repayment is long enough, the optimal strategy tends to be a lump sum debt repayment. Directly paying into the PAYG is the optimal strategy if the investment period is short and the amount invested is relatively small.
 \\
The model presented in this paper could be implemented as an alternative to both the PAYG and mixed pension systems as we are not advocating a particular strategy but rather offer possibilities allowing to reflect the actual market and societal situation. Additionally, our model could raise the public awareness of the financial sustainability of the PAYG as the expected increases in contributions affect every individuals' finances. According to Fornero \cite{Fornero}, if the participants of a pension scheme do not understand its basic principles then the reforms are repealed. Knowledge on pension basic principles is thus important not only for individual's well-being (planning) but also for the society as a whole.\medskip\\
Finally, based on the methodology presented in this paper, at least three important directions for future research can be identified.
First, it would be interesting to explore the smoothed and affordable contributions to be invested by the individual to make the model more applicable in the real world. 
Another direction would be to study the risk sharing between the government and the individuals under different scenarios by using for instance Value at Risk and Expected Shortfall risk measures. Finally, one can analyse a possible structure and asset allocation problems of the fund that might be used in the real life. 
\subsection*{Acknowledgements}
Mar\'ia del Carmen Boado-Penas is grateful for the financial assistance received from the Spanish Ministry of the Economy and Competitiveness [project ECO2015-65826-P]. \smallskip
\\The research of Julia Eisenberg was funded by the Austrian Science Fund (FWF), Project number V 603-N35.
\bibliographystyle{plain}
\bibliography{PAYG}{}
\end{document}